\documentclass{amsart}
\usepackage[utf8]{inputenc}
\usepackage[english]{babel}
\usepackage{amsmath}
\usepackage{amsfonts}
\usepackage{amssymb}
\usepackage{tipa}
\usepackage{mathrsfs}
\usepackage{amsthm}
\usepackage{color}
\usepackage{graphicx}
\usepackage{epstopdf}
\usepackage{enumerate}
\usepackage{balance}

\newtheorem{theorem}{Theorem}
\newtheorem{coro}{Corollary}
\newtheorem{prop}{Proposition}
\theoremstyle{remark}
\newtheorem{ex}{Example}
\newtheorem{rem}{Remark}

\def\Tr{\mathrm{Tr}\,}
\def\Ha{\mathcal{H}}
\def\Ka{\mathcal{K}}
\def\Ce{\mathbb{C}}
\def\ptr{\mathrm{Tr}}
\def\states{\mathfrak{S}}
\def\<{\langle\,}
\def\>{\,\rangle}
\def\pmei{p_{\text{MEI}}}
\def\popt{p_{\text{opt}}}
\DeclareMathOperator{\diag}{\mathit{diag}}
\begin{document}
\title{Conditions for optimal input states for discrimination of quantum channels}
\author{Anna Jen\v cov\'a}
\author{Martin Pl\'avala}
\begin{abstract}
We find optimality conditions for testers in discrimination of quantum channels. These conditions are obtained using semidefinite programming and are similar to optimality conditions for discrimination of quantum states. We get a simple condition for existence of an  optimal tester with any given input state with maximal Schmidt rank, in particular with  a  maximally entangled input state. In case when maximally entangled state is not  optimal an upper bound on the optimal success probability is obtained. The results for discrimination of two channels are applied to covariant channels, qubit channels, unitary channels and simple projective measurements.
\end{abstract}
\maketitle

\section{Introduction}

The problem of multiple hypothesis testing in the setting of quantum channels can be formulated as follows. Assume that $\Phi$ is an unknown quantum channel, but some a priori information is available, in the sense that $\Phi$  is one of given channels $\Phi_1,\dots,\Phi_n$, with probabilities $\lambda_1,\dots,\lambda_n$. The task is to find a procedure that determines the true channel, with the greatest possible probability of success.

For quantum states, this problem was formulated  by Helstrom \cite{helstrom} and since then has been the subject of active research, see e.g. \cite{barnett-crocke_sdp} for an overview and further references.  Here, an ensemble $\{\lambda_i,\rho_i\}_{i=1}^n$ is given, where $\rho_1,\dots,\rho_n$ are quantum states with prior probabilities $\lambda_1,\dots,\lambda_n$, with a similar interpretation as above. A testing procedure, or a measurement, for this problem is described by a positive operator valued measure (POVM) $M$, defined as a collection of positive operators $M_1,\dots, M_n$ summing up to the identity operator $I$. The value $\Tr M_i\rho_j$ is interpreted as the probability that the procedure chooses $\rho_i$ while the true state is $\rho_j$. The task is to maximize the average success probability
\[
p(M)=\sum_i \lambda_i \Tr M_i\rho_i
\]
over all POVMs. In the case  $n=2$, it is well known that the optimal POVM is projection valued, given by the projections onto the positive and negative parts of the operator $\lambda\rho_1-(1-\lambda)\rho_2$, \cite{helstrom}. For $n>2$, there is no explicit expression for the optimal POVM in general, but it is known that a POVM $M$ is optimal if and only if it satisfies 
\begin{equation}\label{eq:optPOVM}
\sum_i \lambda_i\rho_iM_i\ge \lambda_j\rho_j,\qquad \forall j.
\end{equation}
This condition was obtained in \cite{holevo,yuen-kennedy-lax} using the methods of semidefinite programming.

In the case of quantum channels, a most general measurement scheme is described  by a triple $(\Ha_0,\rho,M)$, where $\Ha_0$ is an ancilla, $\rho$ a (pure) state 
on $\Ha\otimes \Ha_0$ and $M=\{M_1,\dots,M_n\}$ is a POVM on $\Ka\otimes \Ha_0$. For $i,j\in\{1,\dots,n\}$, the value
\[
\Tr M_i(\Phi_j\otimes id)(\rho)
\]
is interpreted as the probability that $\Phi_i$ is chosen when the true channel is $\Phi_j$. The average success probability is then
\begin{align}\label{eq:succprob}
p(M,\rho)=\sum_i \lambda_i \Tr M_i(\Phi_i\otimes id)(\rho).
\end{align}
The task is to maximize this value over all triples $(\Ha_0,\rho,M)$.

It was observed  \cite{kitaev,darianopp, sacchi_ent,sacchi_opt} that using entangled input states may give greater success probability and it was shown in \cite{PianiWatrous} that every entangled state is useful for some channel discrimination problem. However, there are situations when e.g. the maximally entangled input state does not give an optimal success probability. It is therefore important to find out whether an optimal scheme with a given input state exists. 

In the broader context of generalized decision problems, conditions for existence of an optimal  scheme with an input state having maximal Schmidt rank were found in \cite{ja_base}, a related problem was studied in \cite{matsumoto}. In the present paper, we show that these conditions can be obtained using the methods of semidefinite programming. Such methods were already applied before in the context of discrimination of quantum channels, see \cite{chiribella_sdp, Watrous-SDP, gutoski}. Compared to these works, we are more concerned with the choice of an optimal input state.
 It is an easy observation that if a given scheme $(\Ha_0, \rho,M)$ is optimal, then $M$ must be an optimal measurement for the ensemble $\{\lambda_i, (\Phi_i\otimes id)(\rho)\}$. We show that, at least in the case that the input state is assumed to have   maximal Schmidt rank, the optimality condition for a channel measurement can be divided into the condition \eqref{eq:optPOVM} for this ensemble and an additional condition that ensures optimality of the input state.
  If the Schmidt rank of the input state is not maximal we obtain a comparably weaker result, but show an example where the use of such an input state is required. 

As an important special case, we get a necessary and sufficient condition for existence of an optimal scheme with a maximally entangled input state. If this  condition  is not fulfilled, we give an upper bound on the optimal success probability. For discrimination of two channels, we use the known form of an  optimal POVM for  two states to obtain a relatively simple condition in terms of Choi matrices of the involved channels, which we call the \eqref{eq:MEI} condition. We also derive an upper bound on the diamond norm, which is tighter than the previously known bound given e.g. in \cite{bph2015generic}, see Remark \ref{rem:bounds} below. The results are  applied to discrimination of covariant channels, qubit channels, unitary channels and simple projective measurements.

The paper is organized as follows: in the next section we  rewrite the problem  as a problem of SDP from which we obtain necessary and sufficient conditions for optimal solution and derive an upper bound on the optimal success probability. In Section \ref{sec:two} we investigate the \eqref{eq:MEI} condition and the related bounds. It the last two sections, we study special cases of channels and  present some examples demonstrating the results.

\section{Optimality conditions}

Let $\Ha$ be a finite dimensional Hilbert space. We denote by  $B(\Ha)^+$ the set of positive operators and by
$\states(\Ha)$ the set of states, that is, positive operators of unit trace. A completely positive trace preserving map $\Phi: B(\Ha)\to B(\Ka)$ is called a channel, we will denote the set of all channels by $\mathcal C(\Ha,\Ka)$. Any linear map $\Phi: B(\Ha)\to B(\Ka)$ is represented
by its Choi matrix $C(\Phi)\in B(\Ka\otimes\Ha)$, defined in \cite{choi} as
\[
C(\Phi):=(\Phi\otimes id)(|\psi_\Ha\>\<\psi_\Ha|),\ |\psi_\Ha \>=\sum_i |i\>\otimes |i\>,
\]
 here $\{|i\>\}$ is a fixed orthonormal basis of $\Ha$. Note that $\Phi\in \mathcal C(\Ha,\Ka)$ if and only if $C(\Phi)$ is positive and $\ptr_\Ka C(\Phi)=I_\Ha$.

An alternative description of a channel measurement is given in terms of process POVMs \cite{ziman} (or  testers \cite{daria_testers}, see also \cite{guwat}). A process POVM is a collection $F=\{F_1,\dots,F_n\}$ of positive operators in $B(\Ka\otimes \Ha)$ with $\sum_i F_i=I\otimes \sigma$ for some state $\sigma\in \states(\Ha)$. For any triple $(\Ha_0,\rho,M)$, there is a process POVM $F$ such that for all  $ \Phi \in \mathcal C(\Ha,\Ka)$ and $i=1,\dots,n$,
\begin{align}\label{eq:ppovm}
\Tr M_i(\Phi\otimes id)(\rho)=\Tr C(\Phi)F_i.
\end{align}
 Conversely, for any process POVM $F$, one can find some ancilla $\Ha_0$, a pure state $\rho$ and a POVM $M$ such that 
(\ref{eq:ppovm}) holds \cite{ziman}. To see this, let $\rho=|\psi\>\<\psi|$, for $\psi\in \Ha\otimes \Ha_0$ and observe that by Schmidt decomposition, we have $|\psi\> = \sum_i \lambda_i (I \otimes U) |i\> \otimes |i\>$ for some unitary operator $U$ and $\lambda_i \geq 0$ for every $i$. Denoting $A = U \sum_i \lambda_i |i\>\<i|$ we get
\begin{equation*}
|\psi\> = (I \otimes A) |\psi_\Ha\>.
\end{equation*}
Since the channel in \eqref{eq:ppovm} acts only on the first part of the system we get
\begin{align}\label{eq:MF}
\Tr  (\Phi\otimes id)(\rho) M_i  =\Tr C(\Phi) (I \otimes A^*) M_i (I \otimes A)
\end{align}
and \eqref{eq:ppovm} holds with $F_i = (I \otimes A^*) M_i (I \otimes A)$. Conversely, let  $\sum_i F_i = I \otimes \sigma$ for $\sigma\in \states(\Ha)$. Let $\sigma^{-1/2}$ be defined on the support  of $\sigma$ and 0 elsewhere, then $M_i = (I \otimes \sigma^{-1/2}) F_i (I \otimes \sigma^{-1/2})$ is a POVM on $\Ka \otimes \Ha_0$ where now $\Ha_0 = \text{supp}(\sigma)$ and \eqref{eq:ppovm} holds as before. 

Using the description by process POVMs, we will show that the maximization of the success probability can be written as a problem of semidefinite programming:
\begin{align*}
\max_{F\in B(\mathbb C^n\otimes \Ka\otimes \Ha)} \Tr CF\\
\mbox{s.t. }\ 
\Tr F&=\dim(\Ka),\\
\Tr (I\otimes X_i )F&=0,\quad i=1,\dots,m,\\
F&\ge 0.
\end{align*}
Here $C=\sum_{i=1}^n |e^n_i\>\<e_i^n|\otimes \lambda_i C(\Phi_i)$, $\{|e_i^n\>\}$ is the canonical basis of $\mathbb C^n$ 
and $X_1,\dots, X_m$ is any basis of the (real) linear subspace 
\[
\mathcal L:=\{X=X^*\in B(\Ka\otimes \Ha), \Tr_\Ka X=0\}.
\]
To see this, note that according to \eqref{eq:ppovm}, equation \eqref{eq:succprob} can be rewritten as:
\begin{equation*}
p(M,\rho)=\sum_i \lambda_i \Tr C(\Phi_i)F_i.
\end{equation*}
Put $F:=\sum_{i=1}^n |e_i^n\>\<e_i^n|\otimes F_i \in B(\Ce^n \otimes \Ka \otimes \Ha)$. Then:
\begin{equation*}
\Tr CF = \sum_i \lambda_i \Tr C(\Phi_i)F_i
\end{equation*}
and the problem of maximizing $p(M, \rho)$ can be understood as the problem of maximizing $\Tr C F$. We have $\ptr_{\Ce^n} F=\sum_i F_i = I \otimes \sigma$ and from $F_i \geq 0$ it follows $F \geq 0$. Note also that since $C$ is block-diagonal, we may extend the maximization over all positive elements $F\in B(\mathbb C^n\otimes\Ka\otimes \Ha)$ with $\ptr_{\Ce^n} F=I\otimes \sigma$, $\sigma\in \states(\Ha)$ (and not only over block-diagonal ones).

To rewrite this to the more usable form  stated above, we need to note that $\ptr_{\Ce^n} F = I \otimes \sigma$ with $\Tr \sigma = 1$ if and only if $\Tr F (I \otimes X) = 0$ for all $X \in \mathcal{L}$ and $\Tr F = \dim(\Ka)$. To prove this statement, let us first assume that $\ptr_{\Ce^n} F = I \otimes \sigma$, then for any $X$,
\begin{align*}
\Tr F (I \otimes X) &= \Tr X \ptr_{\Ce^n} F = \Tr X(I \otimes \sigma) \\
&= \Tr \sigma \ptr_{\Ka} X 
\end{align*}
and $\Tr F = \Tr \ptr_{\Ce^n} F = \Tr I \otimes \sigma = \dim(\Ka)$.

Conversely, assume that $\Tr F (I \otimes X) = 0$ for all $X\in \mathcal{L}$ and $\Tr F = \dim(\Ka)$.  Consider $B(\Ka)$ as a Hilbert space with Hilbert-Schmidt inner product, then there is an orthonormal basis $\{ (\dim(\Ka))^{-1/2}  I,\chi_{\Ka,1},\dots,\chi_{\Ka,N}\}$ in $B(\Ka)$, where each $\chi_{\Ka,j}$ is a self-adjoint operator such that $\Tr \chi_{\Ka,j} = 0$. With respect to this basis, each $X \in \mathcal{L}$ can be expressed as:
\begin{equation*}
X = I\otimes X_{\Ha,0}+  \sum\limits_{j} \chi_{\Ka,j} \otimes X_{\Ha,j}
\end{equation*}
with some $X_{\Ha,i}\in B(\Ha)$. 
From the condition $\ptr_\Ka X=0$ we obtain $X_{\Ha,0}= 0$. Expressing  $\ptr_{\Ce^n} F = I\otimes F_{\Ha,0}+ \sum_j \chi_{\Ka,j} \otimes F_{\Ha,j}$ and using the  condition $\Tr F (I \otimes X) = 0$ we get 
\begin{equation*}
 \Tr  X_{\Ha,j} F_{\Ha,j}  = 0,\qquad \forall j>0.
\end{equation*}
Since there is no restriction on $X_{\Ha,j}$ for $j>0$, we must have $F_{\Ha,j}=0$ for all $j>0$, and hence $\ptr_{\Ce^n} F =   I \otimes F_{\Ha,0}=I\otimes \sigma$. To conclude the proof, from the condition $\Tr F = \dim(\Ka)$ we get $\Tr\sigma  = 1$.
Moreover, it is worth realizing that from the condition $F \geq 0$ we get $\sigma \geq 0$, hence $\sigma \in \states(\Ha)$.

The following result is obtained using standard methods of semidefinite programming (see e.g. \cite{barvinok}). The expression for maximal success probability was obtained also in \cite{chiribella_sdp}, in a more general setting.

\begin{theorem}\label{thm:opt_cond}
Let $\hat F$ be a process POVM. Then $\hat F$ is optimal if and only if there is some $\lambda_0\ge 0$ and some $\Phi_0\in \mathcal C(\Ha,\Ka)$, such that for all $i$,
\[
\lambda_i C(\Phi_i)\le \lambda_0 C(\Phi_0)
\]
and
\[
(\lambda_0C(\Phi_0)-\lambda_iC(\Phi_i))\hat F_i=0,\qquad \forall i.
\]
Moreover, in this case, the maximal success probability is 
\begin{align*}
\Tr \hat F C &=\max_F \Tr FC \\
&=\min_{\Phi\in \mathcal C(\Ha,\Ka)}\min\{\lambda, \lambda_i C(\Phi_i)\le \lambda C(\Phi),\ \forall i\}.
\end{align*}
\end{theorem}
\begin{proof}
As first, we will formulate the dual problem. Let $X_i$, $i=1,\ldots,m$ be some basis of $\mathcal{L}$ and let $y=(y_1,\ldots,y_m) \in \mathbb{R}^m$, then dual problem is:
\begin{align*}
&\min_{\lambda \in \mathbb{R}, y \in \mathbb{R}^m} \lambda \\
\mbox{s.t. }\ &\sum\limits_{i=1}^{m} y_i ( I \otimes X_i ) + \dfrac{\lambda}{\dim(\Ka)} I \geq C.
\end{align*}
Let $\lambda,y_1,\dots,y_m$ be dual feasible, then since $\Tr X_i=0$ and $\Tr C>0$, we must have $\lambda>0$. If we denote $\sum\limits_{i=1}^{m} \dfrac{y_i}{\lambda}  X_i  + \dfrac{1}{\dim(\Ka)} I = :C'$, then
\begin{equation*}
\ptr_\Ka C' = I,
\end{equation*}
and from $(I\otimes C')\ge \lambda^{-1}C\ge0$ we obtain $C'\ge 0$. Hence there is some channel $\Phi \in C(\Ka, \Ha)$, such that $C' = C(\Phi)$. From the condition $\lambda ( I \otimes C(\Phi) ) \geq C$ we obtain
\begin{equation}
\lambda C(\Phi) \geq \lambda_i C(\Phi_i), \label{eq:optCond-ProofIneq}
\end{equation}
for all $i$. From here we see, that the dual problem may be formulated as:
\begin{equation}\label{eq:dual}
\min_{\Phi\in \mathcal C(\Ha,\Ka)}\min\{\lambda, \lambda_i C(\Phi_i)\le \lambda C(\Phi),\ \forall i\}.
\end{equation}

Now let $F' = \frac{1}{n\dim(\Ha)} I$, then $F'$ is a primal feasible plan. Moreover $F'$ belongs to the interior of the cone of positive operators, therefore by Slater's condition we obtain that the duality gap is zero, in other words $\max \Tr CF = \min \lambda$ or $\Tr C \hat{F} = \lambda_0$, where by $\hat{F}$ we denote the primal optimal plan and by $\lambda_0, y_0, \Phi_0$ we denote the dual optimal plan. Since $\hat{F}$ is feasible, we have $\lambda_0 = \sum_{i=1}^{m} y_{0,i} \Tr \hat{F} (I \otimes X_i) + \frac{\lambda_0}{\dim(\Ka)} \Tr \hat{F} = \lambda_0 \Tr ( I \otimes C(\Phi_0)) \hat{F}$ and  we get:
\begin{equation} \label{eq:sum-of-traces-FC}
\sum\limits_i \Tr ( \lambda_0 C(\Phi_0) - \lambda_i C(\Phi_i )) \hat{F}_i = 0.
\end{equation}
As $\lambda_0 C(\Phi_0) - \lambda_i C(\Phi_i ) \geq 0$ and $\hat F_i\ge 0$, the sum may be zero if and only if all summands are zero. Moreover, trace of the product of two positive matrices is zero if and only if their product is zero. To see this let $A, B \geq 0$ and $\Tr (AB) = 0$. We have $\Tr AB = \Tr ((A^{\frac{1}{2}} B^{\frac{1}{2}})^* A^{\frac{1}{2}} B^{\frac{1}{2}} ) = 0$ and hence $A^{\frac{1}{2}} B^{\frac{1}{2}} = 0$ and $AB = 0$.

By the above argumentation, we get  from \eqref{eq:sum-of-traces-FC}
\begin{equation}
\left( \lambda_0 C(\Phi_0) - \lambda_i C(\Phi_i ) \right) \hat{F}_i = 0, \qquad \forall i.\label{eq:optCond-ProofFinal}
\end{equation}

On the other hand, the condition \eqref{eq:optCond-ProofIneq} must hold for any dual feasible plan, but if for some primal and dual feasible plans the condition \eqref{eq:optCond-ProofFinal} holds, then the duality gap for these plans is zero and they are optimal. This concludes the proof.
\end{proof}

Using this result, we can characterize optimality of measurement schemes with input states of maximal Schmidt rank. 
\begin{coro} \label{coro:OptimCond}
Let $\rho\in \states(\Ha\otimes \Ha)$ be a pure state such that $\ptr_1\rho=:\rho_2$ is invertible. Then a measurement scheme $(\Ha,\rho,M)$ is optimal if and only if
\begin{enumerate}[(i)]
\item $Z:=\sum_i\lambda_i (\Phi_i \otimes id)(\rho)M_i $ majorizes $\lambda_i (\Phi_i \otimes id)(\rho)$ for all $i$  \label{item:OptimCondMajor}
\item $\Tr_\Ka Z\propto \rho_2$. \label{item:OptimCondPtr}
\end{enumerate}
\end{coro}
\begin{proof} Let $\rho=|\psi\>\<\psi|$ and let $A\in B(\Ha)$ be the operator such that $|\psi\>=(I\otimes A)|\psi_\Ha\>$, so that the process POVM corresponding to $(\Ha,\rho,M)$ is  given by $\hat F_i=(I\otimes A^*)M_i(I\otimes A)$, see \eqref{eq:MF}.
Note that by our assumptions, $A$ is invertible,  $\rho_2=AA^*$ and $\sum_i \hat F_i=I\otimes A^*A$.

Assume that $\hat F$ is optimal, then by Theorem \ref{thm:opt_cond}, there must be some $\lambda_0>0$ and $\Phi_0\in \mathcal C(\Ha,\Ka)$ such that $\lambda_0C(\Phi_0)\ge \lambda_iC(\Phi_i)$  and 
\begin{equation*}
(\lambda_0 C(\Phi_0)-\lambda_iC(\Phi_i))\hat F_i=0, \qquad \forall i.
\end{equation*}
Summing up over $i$, we obtain
\begin{equation*}
\lambda_0 C(\Phi_0)(I\otimes A^*A) = \sum_i \lambda_i C(\Phi_i ) \hat{F}_i.
\end{equation*}
Multiplying the above equality by $(I\otimes A)$ from the left and by $(I\otimes A^{-1})$ from the right, we get using the above expression for $\hat F_i$,
\begin{align*}
&\lambda_0(I\otimes A)C(\Phi_0)(I\otimes A^*)\\
&=\sum_i\lambda_i(I\otimes A)C(\Phi_i)(I\otimes A^*)M_i =Z.
\end{align*}
The two conditions follow easily from this equality.

Assume conversely that  the conditions \eqref{item:OptimCondMajor}, \eqref{item:OptimCondPtr} are satisfied. Put $Z_0=(I\otimes A^{-1})Z(I\otimes (A^*)^{-1})$, then (i) and (ii) imply that $Z_0\ge 0$ and $\ptr_\Ka Z_0\propto I$. It follows that there is some positive number $\lambda_0$ and $\Phi_0\in \mathcal C(\Ha,\Ka)$ such that $Z_0=\lambda_0C(\Phi_0)$. Moreover, \eqref{item:OptimCondMajor} implies that $\lambda_0C(\Phi_0)\ge \lambda_iC(\Phi_i)$ for all $i$ and 
\begin{align*}
\lambda_0 C(\Phi_0) (I \otimes (A^*A))=   \sum_i \lambda_i C(\Phi_i)\hat F_i.
 \end{align*}
It follows that  $\sum_i(\lambda_0 C(\Phi_0)-\lambda_iC(\Phi_i))\hat F_i=0$ and this implies the optimality condition of Theorem \ref{thm:opt_cond}, exactly as in its proof. 
\end{proof}

Note that \eqref{item:OptimCondMajor} is  the  optimality condition \eqref{eq:optPOVM} for a POVM in discrimination of the ensemble $\{\lambda_i, \rho_i\}$, where $\rho_i=(\Phi_i \otimes id)(\rho)$. In other words, if $\hat M$ is an optimal POVM for this ensemble and 
 \[
 \hat Z:= \sum_i \lambda_i (\Phi_i \otimes id)(\rho) \hat M_i,
 \]
the majorization 
$\hat Z\ge \lambda_i (\Phi_i \otimes id)(\rho)$ is  satisfied. It follows that the existence of an optimal scheme with the given input state is equivalent to the condition (ii). Clearly, in this case, 
$(\Ha,\rho,\hat M)$ is the optimal scheme and  the optimal success probability is
$\popt=\Tr \hat Z$.

Next, we show that the  conditions of Corollary \ref{coro:OptimCond} are necessary for a general pure input state.
\begin{coro} \label{coro:NotFullRankState}
Let $\rho\in \states(\Ha\otimes \Ha)$ be a  pure state such that $\ptr_1\rho=:\rho_2$. Then a measurement scheme $(\Ha,\rho,M)$ is optimal only if the the conditions \eqref{item:OptimCondMajor}, \eqref{item:OptimCondPtr} from the previous corollary hold.
\end{coro}
\begin{proof}
We will show that the measurement scheme is optimal for some problem with reduced input space.
Let us denote by $\Ha_2$ the support of $\rho_2$. Since  $\rho$ is pure,  it must be of the form $\rho = |\psi\>\<\psi| = \sum_{i,j} \sqrt{\xi_i \xi_j} |i \>\< j | \otimes |\tilde{i} \>\< \tilde{j} |$ for some Schmidt decomposition of $|\psi\>$. From here we see that $|\psi\> \in \Ha'_2 \otimes \Ha_2$, where $\Ha'_2$ is a  subspace isomorphic to $\Ha_2$. Let $\Phi_i'$ be the restriction of $\Phi_i$ to $B(\Ha_2')$ and let $P$ be the projection onto $\Ka\otimes \Ha_2$, then it is clear that $\Phi_i'\in \mathcal C(\Ha_2',\Ka)$, moreover, $(\Ha_2,\rho,PMP)$ defines an optimal measurement scheme for the reduced channels, with full Schmidt rank input state. The rest follows from the previous corollary.
\end{proof}

In general, the opposite implication does not hold. That is because if we limit the problem to  some subspace $\Ha_{sub}$ of the original Hilbert space $\Ha$, then in general we don't have a guarantee that the optimal input state will be supported on a subspace of  the form $\Ha_{sub} \otimes \Ha_{anc}$, or in other words we would have to maximize the average success probability over all choices of the subspace $\Ha_{sub}$. We demonstrate this by the following simple example.

\begin{ex}
Let $\rho = |\psi\>\<\psi| \otimes |\varphi\>\<\varphi|$, where $|\psi\>, |\varphi\> \in \Ha$ and let   $M_i = \tilde{M}_i \otimes I$, where $\tilde{M}_i$ is the optimal POVM for discrimination of the ensemble $\{\lambda_i, \Phi_i ( |\psi\>\<\psi| )\}$. By \eqref{eq:optPOVM} we have
\begin{equation*}
\tilde Z = \sum_i \lambda_i \Phi_i (|\psi\>\<\psi|) \tilde{M}_i \geq \lambda_i \Phi_i (|\psi\>\<\psi|)
\end{equation*}
and $Z=\tilde Z\otimes |\varphi\>\<\varphi|$. 
It is easy to see that both conditions \eqref{item:OptimCondMajor} and \eqref{item:OptimCondPtr} are satisfied, but as argued in \cite{PianiWatrous}, there are cases when entangled input states give strictly larger probability of success than any separable state, so that a scheme of the form $(\Ha,\rho,M)$ cannot be  optimal.
\end{ex}

It seems that optimality of  input states  strongly depends on the structure of the channels.  In some cases it is even necessary to use an input state with lower Schmidt rank, because using maximal Schmidt rank input state would "waste" some normalization of the input state on parts of the channels where it is unnecessary, as will be demonstrated in Example \ref{ex:SPM}. It is an open question whether some stronger conditions for general input states can be obtained. See also \cite{sacchi_pauli, sacchi_opt} for a discussion of a similar problem in the case of qubit Pauli channels.

We will next present an upper bound for $\popt$ in the case that condition \eqref{item:OptimCondPtr} is violated. We assume that the input state $\rho$ is maximally entangled, but a similar bound can be obtained  for any input state having a maximal Schmidt rank.

\begin{theorem} \label{thm:ErrorEstimate} Let $M$ be an optimal POVM for discrimination of the ensemble $\{\lambda_i, \dim(\Ha)^{-1} C(\Phi_i) \}$ and let $Z = \sum_i \lambda_i C(\Phi_i) M_i$, $\pmei = \dim(\Ha)^{-1} \Tr Z$. Let $\lVert\cdot \rVert$ denote the operator norm. Then the optimal success probability $\popt$ satisfies
 \[
\pmei \le \popt \le \Vert \ptr_\Ka Z \Vert.
\]
\end{theorem}
\begin{proof} Note that $\pmei$ is the largest success probability that can be obtained by the maximally entangled input state, this implies the first inequality. Further, note that we have $\lambda_iC(\Phi_i) \le Z$ by optimality of the POVM $M$. If now $\lambda>0$ and $\Phi\in \mathcal C(\Ha,\Ka)$ are such that $Z \le \lambda C(\Phi)$, then $\lambda$,$\Phi$ correspond to a dual feasible plan, hence  $\popt \le \lambda$ by \eqref{eq:dual}. To obtain the tightest upper bound in this way, we put
\[
\lambda'_0:=\inf_{\Phi\in \mathcal C(\Ha,\Ka)}\inf\{\lambda>0, Z\le \lambda C(\Phi)\}. 
\]
By the Choi isomorphism, there is some completely positive map $\xi:B(\Ha)\to B(\Ka)$, such that $Z=C(\xi)$. As it was shown in \cite{ja_base} (see Corollary 2 and Section 3.1), $\lambda'_0=\lVert\xi\rVert_\diamond$, where the diamond norm is defined as
\[
\lVert\xi\rVert_\diamond=\sup_{\tau\in\states(\Ha\otimes \Ha)} \lVert(\xi\otimes id) (\tau)\rVert_1.
\] 
Moreover,   since $\xi$ is completely positive, this norm simplifies to
\begin{align*}
\lambda_0'=\lVert\xi\rVert_\diamond&=\sup_{\psi\in \states (\Ha)}\Tr \xi(\psi)= \sup_{\psi\in \states (\Ha)}\Tr Z(I\otimes \psi) \\
&= \sup_{\psi\in \states (\Ha)}\Tr\ptr_\Ka[Z]\psi = \lVert\ptr_\Ka Z \rVert.
\end{align*}
\end{proof}

In general, the bound that we obtain in this way does not have to be meaningful, that is, it may happen that $\Vert \ptr_\Ka Z \Vert > 1$. But, as will be demonstrated by the examples in the last section, there are cases when the bound is meaningful, or even tight.

\begin{rem}\label{rem:epsilon} Note that if $\ptr_\Ka Z=cI$, then $c=p_{MEI}$ and the value of $\epsilon:=\lVert p_{MEI}^{-1} \ptr_\Ka Z-I\rVert$ indicates how much the condition (ii) is violated. It is easy to see that $\lVert \ptr_\Ka Z\rVert \le (1+\epsilon)p_{MEI}$, this shows that if $\epsilon$ is small, the maximally entangled state is close to optimal.\end{rem}

\section{Discrimination of two channels by maximally entangled input states}\label{sec:two}

Let $n=2$ and $\Phi_1,\Phi_2\in \mathcal C(\Ha,\Ka)$. The following notation will be used throughout. Let $\lambda\in (0,1)$, then we put 
\[
\Phi_\lambda=\lambda\Phi_1-(1-\lambda)\Phi_2
\]
and 
\begin{equation}
\Delta_\lambda=\lambda C(\Phi_1)-(1-\lambda)C(\Phi_2)=C(\Phi_\lambda). \label{eq:DeltaLambdaDef}
\end{equation}
Let $\rho=\dim(\Ha)^{-1}|\psi_\Ha\>\<\psi_\Ha|$ be the maximally entangled state and consider any two-outcome POVM on $\Ka\otimes \Ha$,  given by
$\{M,I-M\}$ for some operator $0\le M \le I$ on $\Ka\otimes \Ha$.  The average success probability for the triple
$(\Ha,\rho,M)$  as defined by equation \eqref{eq:succprob} is:
\begin{align*}
p(M, \rho) = \dfrac{1}{\dim(\Ha)} \Tr \Delta_\lambda M + (1 - \lambda ),
\end{align*}
The optimal POVM is obtained if  $M$ is the projection onto the support of the positive part of $\Delta_\lambda$. In this case,  
\begin{align*}
 Z = \sum_i\lambda_i  C(\Phi_i) M_i= (1 - \lambda) C(\Phi_2) + ( \Delta_\lambda )_+
\end{align*}
and  
\[
\pmei = \dim(\Ha)^{-1}\Tr {Z}=\frac{1}{2}(1 + \dim(\Ha)^{-1} \Tr | \Delta_\lambda |).
\]

\begin{coro}\label{coro:maxent} An optimal measurement scheme $(\Ha,\rho,M)$ with a pure maximally entangled input state $\rho$ exists if and only if the Choi operators satisfy 
\begin{equation} \label{eq:MEI}
\Tr_\Ka |\Delta_\lambda|\propto I.\tag{MEI}
\end{equation}
\end{coro}

\begin{proof}
By the remarks below Corollary \ref{coro:OptimCond}, such a scheme exists  if and only if $\Tr_\Ka Z\propto I$, equivalently, $\ptr_\Ka ( \Delta_\lambda )_+ \propto I$. Since we always have  $\ptr_\Ka \Delta_\lambda \propto I$ and 
\[
(\Delta_\lambda)_+=\frac12 (\Delta_\lambda+|\Delta_\lambda|),
\]
 the condition can be rewritten as stated.
\end{proof}

The following corollary describes the upper bound of the optimal probability.
\begin{coro}\label{coro:bounds}
We have the following bounds
\begin{align*}
\pmei \leq \popt \leq \dfrac{1}{2} \left( 1 + \Vert \ptr_\Ka | \Delta_\lambda | \Vert \right).
\end{align*}
If the condition \eqref{eq:MEI} is satisfied, the inequalities become equalities.
\end{coro}
\begin{proof}
We only have to note that if the \eqref{eq:MEI} condition is satisfied, then $\dim(\Ha)^{-1} \Tr \vert \Delta_\lambda \vert = \Vert \ptr_\Ka \vert \Delta_\lambda \vert \Vert$. 
\end{proof}

\begin{rem}\label{rem:bounds} It is well known that $p_{opt}$ is related to the diamond norm as $p_{opt}=\tfrac12(1+\lVert \Phi_\lambda \rVert_\diamond)$. To our knowledge, the only known bounds on the diamond norm in terms of the Choi matrices are the following
\begin{equation} \label{eq:coro-bound-loose}
\dim(\Ha)^{-1}\lVert C(\Phi_\lambda)\rVert_1\le\lVert \Phi_\lambda \lVert_\diamond\le \lVert C(\Phi_\lambda)\rVert_1, 
\end{equation}
(see e.g. \cite[Lemma 6]{bph2015generic}) which is quite coarse. As in Remark \ref{rem:epsilon}, we obtain from Corollary \ref{coro:bounds} the following new upper 
bound:
\begin{align}
\lVert\Phi_\lambda \lVert_\diamond &\le \lVert \ptr_\Ka|C(\Phi_\lambda)|\rVert \nonumber \\
&\le(1+\epsilon')\dim(\Ha)^{-1}\lVert C(\Phi_\lambda)\rVert_1 , \label{eq:coro-bound-epsilon}
\end{align}
where $\epsilon'=\lVert \tfrac {\dim(\Ha)}{\lVert C(\Phi_\lambda)\rVert_1} \ptr_\Ka|C(\Phi_\lambda)|-I\rVert$.  This shows that if \eqref{eq:MEI} is nearly satisfied, the above bounds are quite precise.

To show that the upper bound given by \eqref{eq:coro-bound-epsilon} is better than the bound $\eqref{eq:coro-bound-loose}$ we will show that in general
\begin{equation*}
(1+\epsilon')\dim(\Ha)^{-1} \le 1.
\end{equation*}
We have
\begin{align*}
\epsilon' &=  \dim(\Ha) \left\Vert \dfrac {\ptr_\Ka|C(\Phi_\lambda)|}{\lVert C(\Phi_\lambda)\rVert_1} - \dfrac{1}{\dim(\Ha)} I \right\Vert \\
&\leq \dim(\Ha) -1
\end{align*}
since $\frac {\ptr_\Ka|C(\Phi_\lambda)|}{\lVert C(\Phi_\lambda)\rVert_1}$ is a state. This implies the above inequality. We also see that this inequality is strict unless $\ptr_\Ka|C(\Phi_\lambda)|$ is of rank 1. 
\end{rem}

\section{Applications}
We apply the results of the previous section to the problem of discrimination of covariant channels, unitary channels, qubit channels and measurements. In the case of covariant channels and unital qubit channels, similar results were obtained 
in \cite{matsumoto} for more general decision problems on families of quantum channels.

\subsection{Covariant channels}
Let $\mathcal U(\Ha)$ denote the unitary group of $\Ha$. For 
$U\in \mathcal U(\Ha)$, let  
\[
Ad_U(A):=UAU^*,\quad A\in B(\Ha).
\]
Let $G$ be a group and let $g\mapsto U_g\in \mathcal U(\Ha)$ and $g\mapsto V_g\in \mathcal U(\Ka)$ be  
unitary representations. Assume that $\Phi_1$ and $\Phi_2$ are covariant channels, that is,
\begin{equation} \label{eq:covariant}
\Phi_i\circ Ad_{U_g}=Ad_{V_g}\circ \Phi_i,\qquad i=1,2,\ g\in G.
\end{equation}

Irreducibility of $g\mapsto U_g$ plays a strong role, as we will see. In this case, the only non-zero projection that commutes with all $U_g$ is $I$,  see e.g. \cite{barut-raczka}.

\begin{prop}
Let $\Phi_1, \Phi_2$ be channels satisfying \eqref{eq:covariant}. Assume that the representation $g\mapsto U_g$ is irreducible.  Then the condition \eqref{eq:MEI} is satisfied for any $\lambda\in (0,1)$. 
\end{prop}
\begin{proof}
Let $U^t$ denote the transpose of $U$ with respect to the fixed basis 
$\{|i\>\}$. Let $\lambda\in (0,1)$ and let  $\Delta_\lambda=C(\Phi_\lambda)$ be as in \eqref{eq:DeltaLambdaDef}. We will prove the proposition by showing that $\Tr_\Ka|\Delta_\lambda|$ is invariant under $Ad_{U^t_g}$ and by the discussion above this implies that $\Tr_\Ka|\Delta_\lambda| \propto I$. For every $g\in G$ we have
\begin{align*}
&Ad_{U_g^t}(\Tr_\Ka|\Delta_\lambda|) = \\
&=\Tr_\Ka(id\otimes Ad_{U^t_g})(|\Delta_\lambda|)\\
&=\Tr_\Ka|(\Phi_\lambda\otimes Ad_{U^t_g})(|\psi_\Ha\>\<\psi_\Ha|)|\\
&=\Tr_\Ka|(\Phi_\lambda\circ Ad_{U_g}\otimes id)(|\psi_\Ha\>\<\psi_\Ha|)|\\
&=\Tr_\Ka (Ad_{V_g}\otimes id)(|\Delta_\lambda|)=\Tr_\Ka|\Delta_\lambda|.
\end{align*}

\end{proof}

In case the representation $U$ is reducible, let us sketch an upper bound of the optimal probability. By the previous proof, we have  $Ad_{U_g^t}(\Tr_\Ka|\Delta_\lambda|)$, hence 
\begin{equation*}
\ptr_\Ka | \Delta_\lambda | = \sum\limits_i k_i P^t_i
\end{equation*}
where $k_i \in \mathbb{R}$ and $P_i$ are projections onto the subspaces of the irreducible representation, orthogonal sum of which is $U$. Let $t_i = \Tr P_i$ then $\pmei = \frac{1}{2} (1 + \dim(\Ha)^{-1} \sum_i t_i k_i)$ and we have
\begin{align*}
\popt \leq \dfrac{1}{2} \left( 1 + \max_i k_i \right).
\end{align*}

\subsection{Qubit channels}

Let $\Ha=\Ka=\mathbb C^2$ and let us denote $\psi_\Ha=:\psi_2$. Let $\Gamma(X)=(\Tr X) I - X^t$ be the Werner-Holevo channel, where $X^t$ denotes the transpose map with respect to the canonical basis $|0\>, |1\>$.
Then $\Gamma$ is a unitary channel, given by the unitary $U$ such that 
\[
U|0\>=-|1\>, \qquad U|1\>=|0\>.
\]
It can be easily checked that $\Gamma\circ\Gamma=id$ and  $(id \otimes\Gamma)(|\psi_2\>\<\psi_2|)=(\Gamma\otimes id)(|\psi_2\>\<\psi_2|)$. If $\phi: B(\mathbb C^2)\to B(\mathbb C^2)$ is a linear map such that there is some $a\in \mathbb R$, satisfying
\begin{equation}\label{eq:tracepres}
\Tr\phi(X)=a\Tr X,\qquad X\in B(\mathbb C^2),
\end{equation}
 then 
\[
\phi\circ\Gamma(X)=\Gamma\circ\phi^t(X)+(\Tr X)(\phi(I)-aI),
\]
where $\phi^t(X)=\phi(X^t)^t$. Moreover, for a self-adjoint $X\in B(\mathbb C^2)$, $\Gamma(X)=X^t$ if and only if $X\propto I$. 

Let $\Phi_1$ and $\Phi_2$ be two qubit channels  and let 
$\Delta_\lambda=C(\Phi_\lambda)$ for $\lambda\in (0,1)$ as before. By the previous remarks, the condition \eqref{eq:MEI} 
is equivalent to $\Gamma(\ptr_\Ka|\Delta_\lambda|) =(\ptr_\Ka|\Delta_\lambda|)^t$. We are now going to investigate this equality.

Note that   $\Phi_\lambda$ satisfies (\ref{eq:tracepres}) with $a=2\lambda-1$. Since $\Gamma$ is a unitary channel, we have 
\begin{align*}
\Gamma(\ptr_\Ka|\Delta_{\lambda}|)=&\ptr_\Ka(id\otimes \Gamma)(|\Delta_\lambda|) \\
=&\ptr_\Ka|(\Phi_\lambda\otimes\Gamma)(|\psi_2\>\<\psi_2|)|.
\end{align*}
 Using further properties of $\Gamma$ and $\Phi_\lambda$, we get
\begin{align*}
\ptr_\Ka&|(\Phi_\lambda\otimes\Gamma)(|\psi_2\>\<\psi_2|)\\
=&\ptr_\Ka|(\Phi_\lambda\circ\Gamma\otimes id)(|\psi_2\>\<\psi_2|)|\\
=&\ptr_\Ka|(\Gamma\circ \Phi_\lambda^t\otimes id)(|\psi_2\>\<\psi_2|) \\
&+(\Gamma \circ \Gamma)(\Phi_\lambda(I)-(2\lambda-1)I)\otimes I|\\
=& \ptr_\Ka|(\Phi_\lambda^t\otimes id)(|\psi_2\>\<\psi_2|) \\
&+((2\lambda-1)I-\Phi_\lambda^t(I))\otimes I|\\
=&[\ptr_\Ka|\Delta_\lambda+((2\lambda-1)I-\Phi_\lambda(I))\otimes I|]^t.
\end{align*}
The last equality follows from the fact that $C(\Phi_\lambda^t)=C(\Phi_\lambda)^{t\otimes t}$, where $t\otimes t$ denotes transpose with respect to the product basis $|i\>\otimes |j\>$, and that $\ptr_\Ka|X^{t\otimes t}|=(\ptr_\Ka |X|)^t$ for any 
 $X=X^*\in B(\Ce^2\otimes\Ce^2)$.
Thus we have proved:
\begin{prop} For a pair of qubit channels, the condition \eqref{eq:MEI} holds  if and only if 
\[
\ptr_\Ka|\Delta_\lambda+((2\lambda-1)I-\Phi_\lambda(I))\otimes I| = 
\ptr_\Ka|\Delta_\lambda|.
\]
\end{prop}

In particular, this is true if $\Phi_\lambda (I)=(2\lambda-1)I$. If both channels are unital, this holds for any $\lambda$, hence maximally entangled input state is optimal,  as it was already observed in \cite{matsumoto} and in \cite{sacchi_pauli} in the case of qubit Pauli channels. If $\lambda = \frac{1}{2}$, then the condition \eqref{eq:MEI} is satisfied if $\Phi_1 (I) = \Phi_2(I)$, even if the channels are not unital.

\subsection{Unitary channels}
Let $U,V\in \mathcal U(\Ha)$ and let $\Phi_1=Ad_U$, $\Phi_2=Ad_V$ be the corresponding unitary channels. As it was proved in \cite{darianopp_unitary}, it is not necessary to use entangled inputs for optimal discrimination of two unitary channels.  Nevertheless, it is an interesting question whether a maximally entangled state is also optimal, this will be addressed in this paragraph.

     Let $W = U V^*$. Since any input state $\rho$ may be replaced by $(V^*\otimes I)\rho(V\otimes I)$, it is clear that discrimination of $Ad_U$ and $Ad_V$ is equivalent to discrimination of $Ad_W$ and the identity channel, and that a maximally entangled input state is optimal for one problem if and only if it is optimal for the other. We may therefore assume that $\Phi_1=Ad_W$  and $\Phi_2=id$.
Since the unitaries are given only up to a phase, we may also assume  that  $\Tr W \in \mathbb{R}$.
Put
\begin{align*}
|\phi\>&=\sum_i W|i\>\otimes |i\>, \\
|\psi\>&=\sum_i |i\>\otimes |i\>,
\end{align*}
so that $|\phi\>\<\phi|=C_1$, $|\psi\>\<\psi|=C_2$ are the Choi matrices of the unitary channel $Ad_W$ and identity. By the results of the Appendix it is  clear that $\ptr_1|\Delta_\lambda|\propto I$ if and only if $z\ptr_1(|\phi\>\<\psi|+|\psi\>\<\phi|)\propto I$, where 
\begin{equation*}
z=\<\phi,\psi\>=\Tr  W^*=\Tr W
\end{equation*}
and
\begin{align*}
\ptr_1|\phi\>\<\psi|=   W^t.
\end{align*}
Since the transpose is a linear map and $I^t=I$, we see that \eqref{eq:MEI} is equivalent to 
\begin{equation*}
(\Tr W) (W+W^*)\propto I.
\end{equation*}
If $\Tr W = 0$ this condition is obviously satisfied. If $\Tr W \neq 0$ it is equivalent to
\begin{equation} \label{eq:unit-cond}
W + W^* \propto I.
\end{equation}
The unitary $W$ has a spectral decomposition
\begin{equation*}
W = \sum_{j=1}^{\dim(\Ha)} ( \cos(\alpha_j) + i \sin(\alpha_j) ) | \xi_j \>\< \xi_j |,
\end{equation*}
where  $\alpha_1, \ldots, \alpha_{\dim(\Ha)} \in [-\pi, \pi]$. 
From the condition \eqref{eq:unit-cond} we see that $\cos(\alpha_j)$ must be constant with respect to $j$, or in other words there must exist $\beta \in [0, \pi]$ and numbers $\eta_j \in \{0, 1\}$ such that $\alpha_j = (-1)^{\eta_j} \beta$ for every $j$. By the assumption  $\Tr W \in \mathbb{R}$ we must have
\begin{equation*}
\sum_{j=1}^{\dim(\Ha)} \sin(\alpha_j) = \sin(\beta) \sum_{j=1}^{\dim(\Ha)} (-1)^{\eta_j} = 0.
\end{equation*}
This implies that either $\beta = 0$ and $W = I$, or $W$ has exactly two eigenvalues, each of the same multiplicity. The fact that in our calculation the eigenvalues are complex conjugate of each other is simply caused by the choice $\Tr W \in \mathbb{R}$ and does not have to be generally required. We have proved the following:
\begin{prop} \label{prop:unit-chan} Let $\Phi_1=Ad_U$, $\Phi_2=Ad_V$ be unitary channels.
Put $W=U V^*$ and let $\lambda\in (0,1)$. Then \eqref{eq:MEI} holds if and only if either $\Tr W=0$ or $W$ has at most two different eigenvalues, each of the same multiplicity.
\end{prop}

Note that if $\dim(\Ha)$ is odd, MEI holds iff $\Tr W=\<\psi,\phi\>=0$, in which case the two channels are perfectly distinguishable.

\subsection{Simple projective measurements}

A special case of a channel is a measurement, which is given by a POVM $M=\{M_1,\dots,M_m\}$. One shot discrimination of quantum measurements was investigated in \cite{sed_zim}, where it was proved that entangled input states are necessary in some cases.

The corresponding channel
$\Phi_M: B(\Ha)\to B(\mathbb C^m)$ is defined as 
\[
A\mapsto \sum_i (\Tr M_iA) |i\>\<i|
\]
and the Choi matrix has the form $C(\Phi_M)= \sum_i |i\>\<i|\otimes M_i^t$. Let $\Phi_1=\Phi_{M}$, $\Phi_2=\Phi_N$ for two POVM's  
$M$, $N$ with $m$ outcomes. In this case, the condition MEI has the form
\begin{equation*}
\sum_i|\lambda M_i-(1-\lambda)N_i|\propto I.
\end{equation*}

We will further investigate simple projective measurements.  Let $\{|\xi_i\>\}$ and $\{|\eta_i\>\}
$  be two orthonormal bases in $\Ha$ an let $M_i=P_{\xi_i}:=|\xi_i\>\<\xi_i|$ and $N_i=P_{\eta_i}:=|\eta_i\>\<\eta_i|$. We will also assume that $\lambda=\tfrac12$. 

\begin{prop}
Assume that $P_{\xi_j}=P_{\eta_j}$ for some $j$. Then the condition MEI is satisfied if and only if $M=N$.
\end{prop}

\begin{proof} In this case, the condition is
\[
\sum_i |P_{\xi_i}-P_{\eta_i}|=\sum_{i\ne j}|P_{\xi_i}-P_{\eta_i}|\propto I.
\]
Since $|P_{\xi_i}-P_{\eta_i}|\le I-P_{\xi_j}=I-P_{\eta_j}$  
for all $i\ne j$, this can be true only if $\sum |P_{\xi_i}-P_{\eta_i}|=0$, that is, $M=N$.

\end{proof}

From now on we will always  assume that $|\<\xi_i,\eta_i\>|<1$, equivalently, 
$P_{\xi_i}\ne P_{\eta_i}$, for all $i$.   Then $|P_{\xi_i}-P_{\eta_i}|=c_iP_{\xi_i,\eta_i}$, where $c_i=(1-|\<\xi_i,\eta_i\>|^2)^{1/2}$ and $P_{\xi_i,\eta_i}$ is the projection onto $\mathrm{span}\{\xi_i,\eta_i\}$, so that the condition MEI becomes
\begin{equation}\label{eq:simplemea}
\sum_i c_iP_{\xi_i,\eta_i}=dI, \quad d=2\dim(\Ha)^{-1}\sum_ic_i.
\end{equation}
Note that if $\dim(\Ha)=2$, then $P_{\xi_i,\eta_i}=I$, so that the condition trivially holds. In this case, $\Phi_M$ and $\Phi_N$  are unital qubit channels, so that this follows also 
by previous  results.
Put
\[
|\xi_i^\perp\>=c_i^{-1}(|\eta_i\>-\<\xi_i,\eta_i\>|\xi_i\>), \quad i=1,\dots,\dim(\Ha).
\]
Then $P_i:=P_{\xi_i,\eta_i}= P_{\xi_i,\xi_i^\perp}=P_{\xi_i}+P_{\xi_i^\perp}$. The condition (\ref{eq:simplemea}) is equivalent to
\[
d|\xi_j\>=\sum_ic_iP_i|\xi_j\>=c_j|\xi_j\>+\sum_{i\ne j}c_i\<\xi_i^\perp,\xi_j\>|\xi_i^\perp\>
\]
for all $j$, or
\begin{equation}\label{eq:simplemea2}
(d-c_j)\<\xi_k,\xi_j\>=
\sum_{i}c_i\<\xi_i^\perp,\xi_j\>\<\xi_k,\xi^\perp_i\>,
\end{equation}
for all $j,k$. Note that the diagonal matrix $dI-C$, where $C = \diag(c_1, \ldots, c_n)$, is invertible. Indeed,
$d-c_j=0$ for some $j$ implies that
\[
d-c_j=\sum_{i\ne j} c_i^{-1}|\<\eta_i,\xi_j\>|^2= 0,
\]
so that $\<\eta_i,\xi_j\>=0$ for all $i\ne j$. But then $|\<\eta_j,\xi_j\>|=1$, which is a contradiction. Hence $d-c_j>0$ for all $j$ and $dI-C$ is  positive definite.

Let us begin from \eqref{eq:simplemea2}. Let us denote $W_{ij} = \< \xi_i, \eta_j \>$. Clearly $W$ is a unitary matrix. It is straightforward to see, that
\begin{equation*}
\< \xi_i, \xi^\perp_j \> = (1 - \delta_{ij}) W_{ij} c_j^{-1},
\end{equation*}
so the condition \eqref{eq:simplemea2} becomes
\begin{equation}
(d - c_j) \delta_{kj} = \sum\limits_{i} (1 - \delta_{ki}) W_{ki} c_i^{-1} (1 - \delta_{ij}) W^*_{ij}, \label{SPM-WelementsEq}
\end{equation}
which can be written as a matrix equation of the form
\begin{equation}
dI - C = \left( W - \diag(W) \right) C^{-1} \left( W^{*} - \diag(W^{*}) \right). \label{SPM-WmatrixEq}
\end{equation}

At this point we are ready to settle the case $\dim(\Ha)=3$. 
\begin{prop}
Let $\dim(\Ha) = 3$, then the \eqref{eq:MEI} condition holds if and only if the matrix $W$, defined as above, has of one of the following forms
\begin{align*}
&W_1 =
\begin{pmatrix}
0 & 0 & e^{i\varphi_1} \\
e^{i\varphi_2} & 0 & 0 \\
0 & e^{i \varphi_3} & 0
\end{pmatrix}, \\
&W_2 =
\begin{pmatrix}
0 & e^{i\varphi_1} & 0 \\
0 & 0 &e^{i\varphi_2} \\
e^{i \varphi_3} & 0 & 0
\end{pmatrix},
\end{align*}
for $\varphi_i \in \mathbb{R}$. In other words, one POVM is a cyclic permutation of the other. 
\end{prop}

\begin{proof}
Consider the equation \eqref{SPM-WelementsEq} and let $j \neq k$, then we obtain
\begin{equation}
0 = W_{ki} c_i^{-1} \bar{W}_{ji}, \label{SPM-3Doffdiag}
\end{equation}
for $i \neq j \neq k \neq i$. By putting $i=1,2,3$, we obtain $0= W_{31} \bar{W}_{21}= W_{12} \bar{W}_{23}= W_{23} \bar{W}_{13}$. It follows that some off-diagonal elements of the matrix $W$ must be zero. On the other hand, for $k=j$ the equation \eqref{SPM-WelementsEq} becomes
\begin{equation}
d - c_j = \sum\limits_{i \neq j} W_{ji} c^{-1}_i \bar{W}_{ji} \label{SPM-3Ddiag}
\end{equation}
which shows that some off-diagonal elements must be non-zero. Putting, say, $W_{12}=0$, the above equalities imply that 
also $W_{23}=W_{31}=0$ and all other off-diagonal elements are nonzero. From the condition $WW^*=I$, we obtain that 
$W_{11} = W_{22} = W_{33} = 0$ and $|W_{21}|=|W_{32}|=|W_{13}|=1$, this implies that $W$ is of the form $W_1$ and the basis $|\xi_i\>$ and $|\eta_j\>$ are just cyclic permutations of each other, modulo phase change.
Similarly, assuming that $W_{12}\ne 0$, we obtain that $W$ is of the form $W_2$, which is just the other possible cyclic permutation of the basis, modulo complex phase.

For the converse, it is easy to check that both $W_1$ and $W_2$ satisfy \eqref{SPM-WmatrixEq}. See also the remark following the proof of Proposition \ref{prop:const}.

\end{proof}

The basis vectors $\{|\xi_i\>\}$ and $|\eta_i\>$ are given only up to a phase, note that also the condition \eqref{SPM-WmatrixEq} is invariant under phase changes. We may therefore assume that $\< \xi_i,\eta_i \> \geq 0$ for all $i$. Then the matrix $\diag(W)$ has only real non-negative entries, hence it is self-adjoint.

From $c_i = \sqrt{1-|\< \xi_i, \eta_i \> |^2} = \sqrt{1-W_{ii}^2}$ we can express $W_{ii} = \sqrt{1 - c_i^2}$ and hence $\diag(W) = \sqrt{I-C^2}$. Since both $C$ and $\diag(W)$ are self-adjoint matrices, their functions are well defined by the spectral theorem. Multiplying out the right hand side of the equation \eqref{SPM-WmatrixEq} and using the the obtained expression for $\diag(W)$ we obtain:
\begin{align}
dI =& WC^{-1}W^* + C^{-1} - \left( C^{-1} \sqrt{I - C^2} W^* \right. \nonumber \\
&+ \left. W C^{-1} \sqrt{I - C^2} \right). \label{SPM-WCmatrixEq}
\end{align}
Clearly the operators $WC^{-1}W^* + C^{-1}$ and $C^{-1} \sqrt{I - C^2} W^* + W C^{-1} \sqrt{I - C^2}$ are self-adjoint. Let us denote $\mathscr{L}_S (\mathcal{H})$ the real linear space of self-adjoint operators on $\mathcal{H}$, endowed with a Hilbert-Schmidt scalar product, which we denote as $(\cdot, \cdot)_{HS}$. Let us choose an orthonormal basis on $\mathscr{L}_S$ composed of the operator $\dim(\mathcal{H})^{-\frac{1}{2}} I$ and of operators $X_i$, such that $\Tr(X_i)=0$.

From equation \eqref{SPM-WCmatrixEq} we can see that:
\begin{align}
WC^{-1}W^* + C^{-1} = 
\beta \dim(\mathcal{H})^{-\frac{1}{2}} I + L, \label{SPM-Wsep1} \\
C^{-1} \sqrt{I - C^2} W^* + W C^{-1} \sqrt{I - C^2} = \nonumber \\
b \dim(\mathcal{H})^{-\frac{1}{2}} I + L, \label{SPM-Wsep2}
\end{align}
where $L$ is some real linear combination of $X_i$, hence $\Tr L=0$ and $\beta, b \in \mathbb{R}$. It is easy to compute $\beta$ as follows:
\begin{align*}
\beta &= (\dim(\mathcal{H})^{-\frac{1}{2}} I, WC^{-1}W^* + C^{-1} )_{HS} \\
&= 2 \Tr\left( \dim(\mathcal{H})^{-\frac{1}{2}} C^{-1} \right) \nonumber \\
&= 2 \dim(\mathcal{H})^{-\frac{1}{2}} \sum\limits_i c^{-1}_i.
\end{align*}
Now it is easy to see, that:
\begin{equation*}
b = 2 \dim(\mathcal{H})^{-\frac{1}{2}} \sum\limits_i \left( c^{-1}_i - c_i  \right).
\end{equation*}
because from equation \eqref{SPM-WCmatrixEq} we see that $\beta - b = 2 \dim(\mathcal{H})^{-\frac{1}{2}} \sum_i c_i$ as $d = 2 (\dim(\mathcal{H}))^{-1} \sum_i c_i$.

We have proved the following:
\begin{prop}
Let $M_i = P_{\xi_i}$ and $N_i = P_{\eta_i}$ for some orthonormal bases $|\xi_i\>$ and $|\eta_j\>$. Then \eqref{eq:MEI} condition is equivalent to the equations \eqref{SPM-Wsep1}, \eqref{SPM-Wsep2}.
\end{prop}

The obtained equations look rather complicated, but they yield some results in specific cases.
\begin{prop}\label{prop:const}
With the above notations, assume that $c_i = c$ for all $i$. Then
\begin{enumerate}
\item[(i)] If $c\neq 1$ and $\dim(\Ha)$ is odd, then \eqref{eq:MEI} cannot be satisfied.
\item[(ii)]  If $c\neq 1$ and $\dim(\Ha)$ is even, then  \eqref{eq:MEI} is satisfied if and only if $W = \sqrt{1-c^2}I + icG$, where $G$ is a symmetric unitary matrix with zero diagonal.
\item[(iii)] If $c=1$, then  \eqref{eq:MEI}  always holds.

\end{enumerate}
\end{prop}
\begin{proof}
Let's assume that $c_i = c \neq 1$ for all $i$.  Immediately we see, that $C^{-1} = c^{-1} I$ and $\sqrt{I - C^2} = \sqrt{1 - c^2} I$. As a direct consequence of this we obtain $L=0$ from the equation \eqref{SPM-Wsep1}. Using this to our advantage in the equation \eqref{SPM-Wsep2} we obtain that in this case, \eqref{eq:MEI} is equivalent to
\begin{equation}
W^* + W = 2 \sqrt{1-c^2} I. \label{SPM-GoMUB-W+W*}
\end{equation}
As in the proof of Proposition \ref{prop:unit-chan}, we see that $W$ has exactly two eigenvalues with the same multiplicity, this also implies (i).  If $\dim(\Ha)$ is even, it is clear that these eigenvalues must be equal to $\lambda_\pm:=\sqrt{1-c^2}\pm ic$. Let $P$ be the eigenprojection corresponding to $\lambda_+$, then $W=\sqrt{1-c^2}I+icG$,  with $G=2P-I$. Conversely, it is easy to see that if $W$ is of this form, then \eqref{SPM-GoMUB-W+W*} holds, this finishes the proof of (ii).

Now assume that $c_i = 1$, which means that $\langle \xi_i, \eta_i \rangle = 0$ for all $i$. Again we immediately see that conditions 
\eqref{SPM-Wsep1} and \eqref{SPM-Wsep2} are satisfied, with  $L=0$. This implies (iii). 

\end{proof}

 Note that  we can conclude from the last statement that any two bases, such that the unitary matrix mapping one basis to the other is hollow satisfy the condition \eqref{eq:MEI}. Especially, this happens if the basis  $|\eta_j\rangle$ is  a permutation of $|\xi_j\rangle$, leaving no element fixed.

A particular case of the situation described in Proposition \ref{prop:const} is when the bases are mutually unbiased, that is when $|\<\xi_i,\eta_j\>|=\dim(\Ha)^{-1/2}$ for all $i,j$. Then $H = \sqrt{dim(\Ha)} W$ is a Hadamard matrix, see e.g. \cite{Tadej2006} for more information on complex Hadamard matrices. As we have seen, such bases can satisfy the condition \eqref{eq:MEI} only in even-dimensional Hilbert spaces. Below, we provide a further result for $\dim(\Ha)=4$. Recall that two Hadamard matrices are equivalent if one can be turned into the other by multiplication by diagonal unitaries and permutations.  Up to equivalence, any Hadamard matrix can be turned into a dephased form, with all elements in the first row and column equal to unity. In this way, any four-dimensional Hadamard matrix is equivalent to a member of a one-parameter family, containing a unique matrix $H_{\mathbb R}$ with real entries. Since the vectors $|\xi_i\>$ and $\eta_i\>$ are given only up to a phase,  
the next result shows that this is the only case when  \eqref{eq:MEI} is satisfied.

\begin{prop}
Let $\dim(\Ha) = 4$, and let the bases $|\xi_j\>$ and $|\eta_j\>$ be mutually unbiased. Then the condition \eqref{eq:MEI} is satisfied if and only if the corresponding Hadamard matrix satisfies $2W=D_1H_{\mathbb R}D_2$, where $D_1$, $D_2$ are diagonal unitaries and
 \begin{equation*}
H_{\mathbb R}= \begin{pmatrix}
1 & 1 & 1 & 1 \\
1 & -1 & -1 & 1 \\
1 & 1 & -1 & -1 \\
1 & -1 & 1 & -1
\end{pmatrix}.
\end{equation*}
\end{prop}

\begin{proof} It is clear that $c_i=c=\sqrt{3}/2$. By Proposition \ref{prop:const}, we see that we need to search for a $4 \times 4$ hollow symmetric unitary matrix $G$ such that each off-diagonal element has modulus $1/\sqrt{3}$. A general form for such a matrix can be found by a straightforward  computation. We find that
\begin{equation}\label{eq:2W}
2W= I+i\sqrt{3}G= \begin{pmatrix}
1 & \bar{b}c & i \bar{b} & i \bar{a} \\
-b\bar{c} & 1 & i \bar{c} & -i\bar{a} b \bar{c} \\
ib & ic & 1 & \bar{a} b \\
ia & -ia\bar{b}c & -a\bar{b} & 1
\end{pmatrix}.
\end{equation}
for $|a|=|b|=|c|=1$, so that $2W=
D_1H_{\mathbb R}D_2$, where $D_1=diag(1,\bar bc, i\bar b, i\bar a)$ and 
$D_2=diag(1,-b\bar c, ib, ia)$. 

Conversely, suppose that $2W=D_1H_{\mathbb R}D_2$ for diagonal unitaries $D_1$ and $D_2$. Since we also assume that all diagonal elements of $W$ are equal to 1, we obtain that we must have $D_1=diag(d_1,d_2,d_3,d_4)$ and $D_2=diag(\bar d_1, -\bar d_2,-\bar d_3,-\bar d_4)$. It is now easy  to check that $2W$ has the form \eqref{eq:2W}, with 
$a=-i\bar d_1d_4$, $b=-i\bar d_1d_3$ and $c=i\bar d_2d_3$. 

\end{proof}

\section{Examples}
In this section, we present examples based on the results of the previous sections.

\subsection{Qubit channels}
To underline how maximally entangled input states may be used for discrimination of two qubit channels, we will present an example of discrimination of identity channel $\phi_{id}$ and amplitude damping channel $\phi_{AD}$. The amplitude damping channels are not unital and hence \eqref{eq:MEI} is not necessarily satisfied, and we will see that the maximally entangled state is indeed not optimal. The identity channel may be replaced by a unitary channel with some changes to the following calculations. We will set $\lambda=\frac{1}{2}$.

\begin{figure}
\begin{center}
\includegraphics[width=\linewidth]{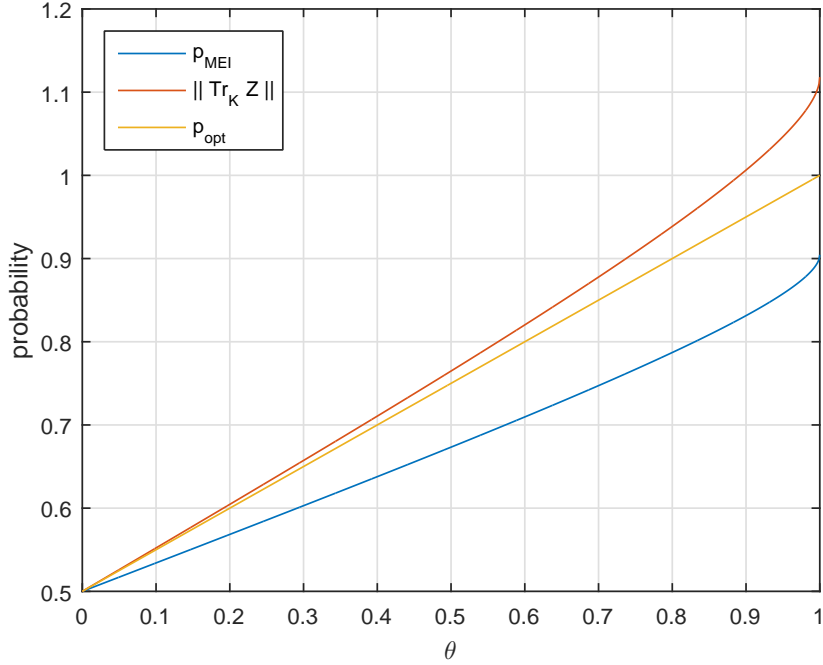}
\caption{The dependence of success probability with maximally entangled input state $\pmei$, optimal  success probability $\popt$ and the upper bound $\Vert \ptr_\Ka Z \Vert$ given by Thm. \ref{thm:ErrorEstimate} on the parameter $\theta$ in discrimination of the amplitude damping channel $\phi_{AD}$   and the identity channel (Example \ref{ex:qubit}).} \label{fig:prob-theta}
\end{center}
\end{figure}

\begin{ex} \label{ex:qubit}
Let $\Ha$ denote complex Hilbert space with $\dim(\Ha)=2$ and $|0\>, |1\>$ will denote some orthonormal basis of $\Ha$. The amplitude damping channel is represented by Kraus operators $A_\theta$, $B_\theta$, defined as
\begin{align*}
A_\theta &= |0\>\<0| + \sqrt{1-\theta} |1\>\<1|, \\
B_\theta &= \sqrt{\theta} |0\>\<1|,
\end{align*}
where $\theta \in [0,1]$ is a parameter. Note that for $\theta = 0$ the amplitude damping channel becomes the identity. We will proceed as follows: we will find the Choi matrices corresponding to $\phi_{id}$ and $\phi_{AD}$ to obtain $\pmei$ and the upper bound on the optimal success probability. Then we will find $\popt$ to compare it with the upper bound.

Let $C_{id} = C(\phi_{id})$ and $C_{AD} = C(\phi_{AD})$. We will be interested in the matrix $\Delta_\frac{1}{2} = \dfrac{1}{2} ( C_{id} - C_{AD} )$ which is of the form
\begin{equation*}
\Delta_\frac{1}{2} = \dfrac{1}{2} \begin{pmatrix}
0 & 0 & 0 & 1-\sqrt{1 - \theta} \\
0 & 0 & 0 & 0 \\
0 & 0 & -\theta & 0 \\
1-\sqrt{1 - \theta} & 0 & 0 & \theta
\end{pmatrix}.
\end{equation*}
Since $\pmei = \frac{1}{2}(1 + \frac{1}{2} \Tr | \Delta_\frac{1}{2}|)$ we can already find $\pmei$ as a function of $\theta$. It is easy to see that $0$ and $- \frac{1}{2} \theta$ are eigenvalues of $\Delta_\frac{1}{2}$. Finding the other two eigenvalues is easy since the problem reduces to finding eigenvalues of a $2 \times 2$ matrix, they are $\frac{1}{4} \big( \theta + \sqrt{\theta^2 + 4( 1 - \sqrt{1-\theta} )^2 } \, \big)$ and $\frac{1}{4} \big( \theta - \sqrt{\theta^2 + 4( 1 - \sqrt{1-\theta} )^2 } \, \big)$. Let us denote
\begin{align*}
&\lambda_1 = \frac{1}{2} \big( \theta + \sqrt{\theta^2 + 4( 1 - \sqrt{1-\theta} )^2 } \, \big) \\
&\lambda_2 = \frac{1}{2} \big( \theta - \sqrt{\theta^2 + 4( 1 - \sqrt{1-\theta} )^2 } \, \big)
\end{align*}
then we have $\pmei = \frac{1}{2}( 1 + \frac{1}{2} \lambda_1)$. Moreover after a tedious calculation it can be seen that
\begin{align*}
\ptr_1 | \Delta_\frac{1}{2} | =& \dfrac{\lambda_1}{2} I - \dfrac{\lambda_1^2 + \lambda_1 \lambda_2}{2(\lambda_1 - \lambda_2)} |0\>\<0| \\
&+ \dfrac{\lambda_1^2 + \lambda_1 \lambda_2}{2(\lambda_1 - \lambda_2)} |1\>\<1|
\end{align*}
and we have $\popt \leq \frac{1}{2} \Vert I + \ptr_1 | \Delta_\frac{1}{2} | \Vert$. Since $\lambda_1 \geq 0$ and $\lambda_1 \pm \lambda_2 \geq 0$ we have
\begin{equation*}
\popt \leq \dfrac{1}{2} + \dfrac{\lambda_1}{4} \left( 1 + \dfrac{\theta}{\sqrt{\theta^2 + 4( 1 - \sqrt{1-\theta} )^2 }} \right)
\end{equation*}

To verify the upper bound will find $\popt$, for which we only need to find the value of $\frac{1}{2}\Vert \phi_{id} - \phi_{AD} \Vert_\diamond$, which we have done by numeric methods. According to \cite{Watrous-SDP} the problem of finding a diamond norm can be formulated as following SDP problem:
\begin{align*}
\max_{X} \dfrac{1}{2} ( \Tr ( \Delta_{\frac{1}{2}}^* X ) + \Tr ( \Delta_{\frac{1}{2}} X ) )\\
\mbox{s.t. }\ 
\begin{pmatrix}
I \otimes \rho_0 & X \\
X^* & I \otimes \rho_1
\end{pmatrix} &\geq 0,
\end{align*}
where $X \in B(\Ha \otimes \Ha)$, $\rho_0, \rho_1 \in \states(\Ha)$. We used MATLAB and the package CVX for solving the convex program, \cite{cvx, Grant-Boyd}. We have computed the numerical values for $100$ values of the parameter $\theta$, homogeneously distributed on the interval $[0,1]$. Between these points straights lines were drawn, hence the figure looks like a continuous line.

The resulting expressions of $\pmei$, $\Vert \ptr_\Ka Z \Vert$ and numerical data of $\popt$ as functions of $\theta$ are plotted in the Fig. \ref{fig:prob-theta}. Even though it shows that maximally entangled input state is not optimal, notice that the upper bound is close to $p_{opt}$ for small values of the parameter $\theta$.
\end{ex}

\subsection{Unitary channels}

\begin{figure}
\begin{center}
\includegraphics[width=\linewidth]{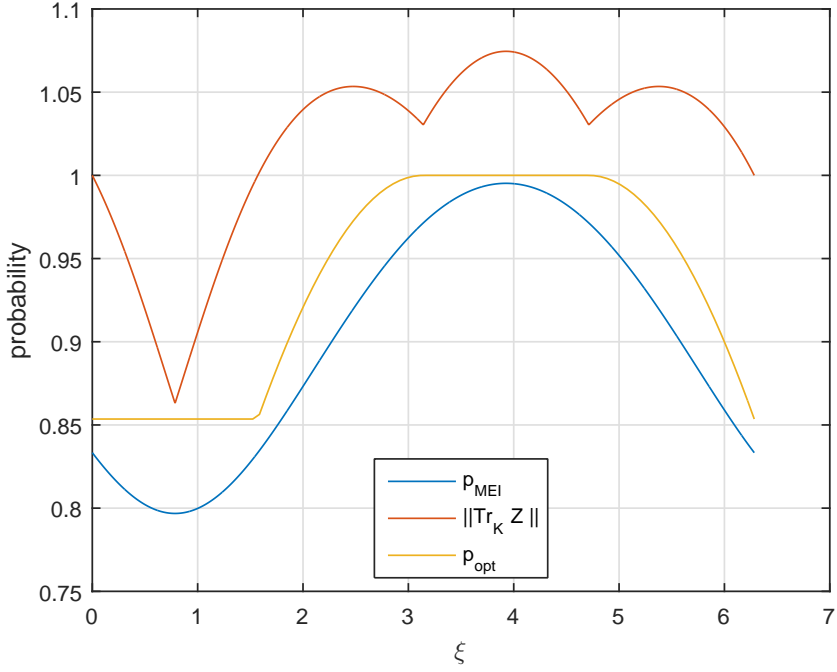}
\caption{The dependence of success probability with maximally entangled input state $\pmei$, optimal success probability  $\popt$ and the upper bound $\Vert \ptr_\Ka Z \Vert$ given by Thm. \ref{thm:ErrorEstimate} on the parameter $\xi$ in 
discrimination of the unitary channel  \eqref{eq:unitChan-W1} and the identity
(Example \ref{ex:unitary}).} \label{fig:prob-xi}
\end{center}
\end{figure}

\begin{ex}\label{ex:unitary}  Let $\lambda = 1/2$ and $\dim(\Ha) = 3$. Without loss of generality, we may always assume that we discriminate the identity channel against a unitary $Ad_W$. As our first example, we consider the unitary matrix 
\begin{equation}
W_1 = \begin{pmatrix}
1 & 0 & 0 \\
0 & i & 0 \\
0 & 0 & e^{i \xi}
\end{pmatrix}. \label{eq:unitChan-W1}
\end{equation}
At this point, it is easy to compute $\pmei$ and $\Vert \ptr_\Ka Z \Vert$, as we can express them as functions of $W_1$, which is only dependent on $\xi$, so the bounds  are functions of $\xi$. Moreover it is possible to compute $\popt$ by numerical methods described in Example \ref{ex:qubit}. The expressions for $\pmei$ and $\Vert \ptr_\Ka Z \Vert$ are long and messy, they are plotted in the Fig. \ref{fig:prob-xi} as well as the computed values of $\popt$. From the figure it is once again clear that on one hand as $\pmei$ rises the the upper bound becomes meaningless, but on the other hand if $\pmei$ is small then the upper bound and $\popt$ is small as well.

Another case we consider is matrix of the form
\begin{equation}
W_2 = \begin{pmatrix}
1 & 0 & 0 \\
0 & \frac{1}{\sqrt{2}} (1 + i) & 0 \\
0 & 0 & e^{i \xi}
\end{pmatrix}. \label{eq:unitChan-W2}
\end{equation}
Again, it is straightforward to obtain $\pmei$ and $\Vert \ptr_\Ka Z \Vert$ and to compute $\popt$ numerically, the obtained functions are plotted in Fig. \ref{fig:prob-xi-2}. Again it can be nicely seen that as $\pmei$ rises towards $1$ the upper bound becomes  meaningless.

\begin{figure}
\begin{center}
\includegraphics[width=\linewidth]{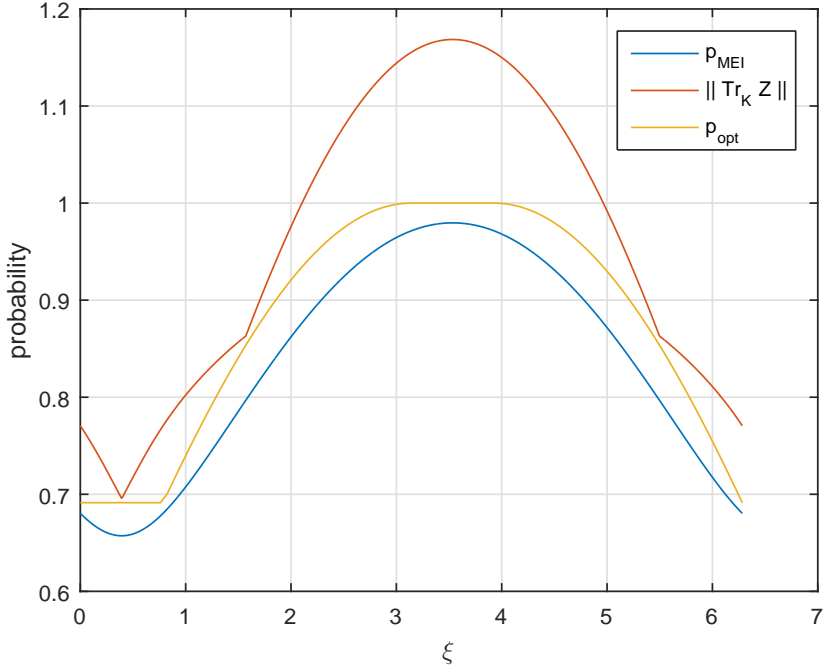}
\caption{The dependence of success probability with maximally entangled input state $\pmei$, optimal success probability  $\popt$ and the upper bound $\Vert \ptr_\Ka Z \Vert$ given by Thm. \ref{thm:ErrorEstimate} on the parameter $\xi$ in 
discrimination of the unitary channel  \eqref{eq:unitChan-W2} and the identity
(Example \ref{ex:unitary}).} \label{fig:prob-xi-2}
\end{center}
\end{figure}
\end{ex}

\subsection{Simple projective measurements}
Let us demonstrate our results once more, this time for simple projective measurements.
\begin{ex} \label{ex:SPM}
Let $\lambda = \frac{1}{2}$, $\dim(\Ha) \geq 3$, $|\eta_1\> = \frac{1}{\sqrt{2}}( |\xi_1\> + |\xi_2\> )$, $|\eta_2\> = \frac{1}{\sqrt{2}}( |\xi_1\> - |\xi_2\> )$, $|\eta_j\> = |\xi_j\>$ for $j \geq 3$. According to our previous results the maximally entangled input state is not optimal for discriminating these two measurements. As first we need to find $\ptr_\Ka \vert \Delta_\frac{1}{2} \vert = \frac{1}{2} \sum_i | P_{\xi_i} - P_{\eta_i}|$. We get
\begin{align*}
| P_{\xi_1} - P_{\eta_1}| &= \dfrac{1}{2} | P_{\xi_1} - P_{\xi_2} - |\xi_1\>\< \xi_2| - |\xi_2\>\< \xi_1| |, \\
| P_{\xi_2} - P_{\eta_2}| &= \dfrac{1}{2} | P_{\xi_2} - P_{\xi_1} + |\xi_1\>\< \xi_2| + |\xi_2\>\< \xi_1| |,
\end{align*}
from which we see, that $| P_{\xi_1} - P_{\eta_1}| = | P_{\xi_2} - P_{\eta_2}|$. The operator $P_{\xi_2} - P_{\eta_2}$ is diagonal in the orthonormal basis $|\varphi_i\>$, $i=1,\ldots, \dim(\Ha)$, where
\begin{align*}
|\varphi_1\> &= \dfrac{1}{\sqrt{2}\sqrt{2-\sqrt{2}}} \left( |\xi_1\> + (1-\sqrt{2}) |\xi_2\> \right), \\
|\varphi_2\> &= \dfrac{1}{\sqrt{2}\sqrt{2+\sqrt{2}}} \left( |\xi_1\> + (1+\sqrt{2}) |\xi_2\> \right), \\
|\varphi_j\> &= |\xi_j\>,
\end{align*}
where $j \geq 3$. In this basis it holds that
\begin{equation*}
P_{\xi_1} - P_{\eta_1} = \dfrac{1}{\sqrt{2}} P_{\varphi_1} - \dfrac{1}{\sqrt{2}} P_{\varphi_2}.
\end{equation*}
The following calculation is straightforward. We obtain
\begin{equation*}
\ptr_\Ka | \Delta_\frac{1}{2} | = \dfrac{1}{\sqrt{2}} (P_{\varphi_1} + P_{\varphi_2}),
\end{equation*}
Now it is easy to see that
\begin{align*}
\pmei &= \dfrac{1}{2} + \dfrac{1}{\sqrt{2} \dim(\Ha)}, \\
\Vert \ptr_\Ka Z \Vert &= \dfrac{2 + \sqrt{2}}{4} \approx 0.8535 \ldots
\end{align*}
i.e. the bound is meaningful and the same for all $\dim(\Ha)$, even though $\pmei$ tends to $\frac{1}{2}$ from above in the formal limit $\dim(\Ha) \rightarrow \infty$.

To underline the correctness of the upper bound we will find the optimal state for discrimination of the channels. Let us denote $\Phi_\xi$, $\Phi_\eta$ the channel corresponding to simple projective measurement corresponding to the set of projectors $\{P_{\xi_i}\}$, $\{P_{\eta_i}\}$ respectively and denote $P_2 = P_{\xi_1} + P_{\xi_2}$. Notice that the channels $\Phi_\xi$ and $\Phi_\eta$  can be separated as follows
\begin{align*}
\Phi_\xi (\rho) &= \sum\limits_{i = 1}^2 \Tr (P_{\xi_i} \rho) P_{\xi_i} + \sum\limits_{i = 3}^{\dim(\Ha)} \Tr (P_{\xi_i} \rho) P_{\xi_i} \\
&= \phi_\xi (\rho) + \chi (\rho), \\
\Phi_\eta (\rho) &= \sum\limits_{i = 1}^2 \Tr (P_{\eta_i} \rho) P_{\xi_i} + \sum\limits_{i = 3}^{\dim(\Ha)} \Tr (P_{\xi_i} \rho) P_{\xi_i} \\
&= \phi_\eta (\rho) + \chi (\rho),
\end{align*}
where $\rho$ is a state on $\Ha$ and $\phi_\xi = \sum_{i = 1}^2 \Tr (P_{\xi_i} \rho) P_{\xi_i}$, $\phi_\eta = \sum_{i = 1}^2 \Tr (P_{\eta_i} \rho) P_{\xi_i}$ and $\chi = \sum_{i = 3}^{\dim(\Ha)} \Tr (P_{\xi_i} \rho) P_{\xi_i}$. Notice that
\begin{align*}
\phi_\xi (\rho) &= \phi_\xi (P_2 \rho P_2), &\phi_\xi (P_2^\perp \rho P_2^\perp) = 0, \\
\phi_\eta (\rho) &= \phi_\eta (P_2 \rho P_2), &\phi_\eta (P_2^\perp \rho P_2^\perp) = 0, \\
\chi (\rho) &= \chi (P_2^\perp \rho P_2^\perp), &\chi (P_2 \rho P_2) = 0,
\end{align*}
where $P_2^\perp = I - P_2$.
This leads to the following
\begin{align*}
\Phi_\xi (\rho) = \Phi_\xi (P_2 \rho P_2 + P_2^\perp \rho P_2^\perp), \\
\Phi_\eta (\rho) = \Phi_\eta (P_2 \rho P_2 + P_2^\perp \rho P_2^\perp).
\end{align*}
In other words, it is sufficient to consider only input states of the form $\rho = \lambda' \rho_2 + (1 - \lambda') \rho_\perp$, where $\rho_2$ and $\rho^\perp$ are states, such that $P_2 \rho_2 P_2 = \rho_2$,  $P_2^\perp \rho_\perp P_2^\perp = \rho_\perp$ and $0 \leq \lambda' \leq 1$. Since
\begin{align*}
\popt =& \max \left( \dfrac{1}{2} \Tr ( (\Phi_\xi \otimes id ) (\rho) M) \right. \\
&+ \left. \dfrac{1}{2} \Tr ( (\Phi_\eta \otimes id ) (\rho) (I - M)) \right) \\
=& \max \left( \dfrac{\lambda}{2} \Big( 1 +  \Tr \big( ( (\phi_\xi \otimes id ) (\rho_2) \right. \\
&- \left. (\phi_\eta \otimes id ) (\rho_2) )_+ \big) \Big) + \dfrac{1-\lambda}{2} \right),
\end{align*}
where $M$, $I - M$ is the POVM we use for discrimination of the channels, it is obvious that for $\popt$ to be maximal we also have to set $\lambda = 1$ and according to corollary \ref{coro:NotFullRankState} the problem reduces to discriminating the channels $\phi_\xi$ and $\phi_\eta$. Both $\phi_\xi$ and $\phi_\eta$ are unital qubit channels and by our previous results the optimal input state $\rho_{opt}$ is
\begin{equation*}
\rho_{opt} = \dfrac{1}{2} \sum\limits_{i,j = 1}^2 |\xi_i\>\<\xi_j| \otimes |\xi_i\>\<\xi_j|.
\end{equation*}
In the product basis given by $|\xi_1\>$, $\xi_2\>$, we have
\begin{align*}
(\phi_\xi \otimes id ) (\rho_{opt}) - (\phi_\eta \otimes id ) (\rho_{opt}) = \\
\dfrac{1}{4} \begin{pmatrix}
1 & -1 & 0 & 0 \\
-1 & -1 & 0 & 0 \\
0 & 0 & -1 & 1 \\
0 & 0 & 1 & 1 \\
\end{pmatrix}
\end{align*}
 This matrix has two eigenvalues $\frac{1}{2\sqrt{2}}$ and $-\frac{1}{2\sqrt{2}}$ each with multiplicity $2$. We get $\Tr \big( ( (\phi_\xi \otimes id ) (\rho_2) - (\phi_\eta \otimes id ) (\rho_2) )_+ \big) = \frac{1}{\sqrt{2}}$ and
\begin{equation*}
\popt = \dfrac{1}{2} \left( 1 + \dfrac{1}{\sqrt{2}} \right) = \dfrac{2 + \sqrt{2}}{4}
\end{equation*}
which is exactly the same as our upper bound.
\end{ex}

\section{Conclusions}
We presented necessary and sufficient conditions for optimality of a process POVM in channel discrimination, especially for a process POVM corresponding to a measurement scheme with full Schmidt rank input state. In particular, a necessary and sufficient condition for existence of an optimal measurement scheme with a given full Schmidt rank input state were found. 
In the case of maximally entangled input states, an upper bound of the optimal success probability was given if the optimality condition is not satisfied. For discrimination of two channels, we obtained a simple condition in terms of the Choi matrices of the channels and a new upper bound on the diamond norm.

The results were applied to discrimination of four types of channels. For covariant channels, known results for the 
irreducible case were extended to some reducible cases and an upper bound on the optimal success probability was found.
For qubit channels, the obtained condition generalizes previously known results to some pairs of  non-unital channels. 
We proved that for discrimination of unitary channels, maximally entangled input states are optimal only  in some very special cases. To our best knowledge, the results obtained for unitary channels and simple projective measurements are new.

An interesting  open question  is whether it is possible to obtain a similar condition for input states with lower Schmidt rank. As it was shown, there are cases when such input states  are optimal and there may even be no optimal full rank input states. It  is not only of question what the  Schmidt rank of the optimal input state may be but also how to select the subspace of the input Hilbert space that will form the support of the partial trace of the input state. Another possible future directions of research is to investigate 
optimal discrimination of more complex quantum processes.

\section*{Acknowledgments}
We thank the anonymous referee for the valuable feedback and comments that made this paper more readable and better organized and also for pointing out an easier proof of Prop. \ref{prop:unit-chan}. This research was supported by grant VEGA 2/0069/16.

\section*{Appendix: Computation of the absolute value of the difference of positive rank-1 operators}
Let $|\phi\>,|\psi\>$ be  unit vectors,  $P_\phi=|\phi\>\<\phi|$, $P_\psi=|\psi\>\<\psi|$. Let $\lambda\in (0,1)$  and 
\[
D_\lambda:=\lambda P_\psi-(1-\lambda)P_\phi.
\]
Then
\[
|D_\lambda|= |\mu_1|P_{\xi_1}+|\mu_2|P_{\xi_2},
\]
where $\mu_1,\mu_2\in \mathbb R$ are the eigenvalues and $\xi_1,\xi_2$ the 
corresponding eigenvectors of $D_\lambda$. Moreover, we have
\[
\xi_i=\alpha_i\phi+\beta_i\psi,\qquad i=1,2
\]
and since we do not worry about a phase, we may suppose that $\alpha_i\ge 0$. 
From  $D_\lambda|\xi_i\>=\mu_i|\xi_i\>$, we obtain
\begin{align*}
\lambda(\alpha_i\bar z+\beta_i)|\psi\>-(1-\lambda)(\alpha_i+\beta_i z)|\phi\>= \\
\mu_i\alpha_i|\phi\>+\mu_i\beta_i|\psi\>,
\end{align*}
where $z=\<\phi,\psi\>$. It follows that 
$\lambda(\alpha_i\bar z +\beta_i)=\mu_i\beta_i$, so that 
\[
\beta_i=k_i \bar z, \qquad i=1,2
\]
where $k_i=\frac{\lambda}{\mu_i-\lambda}\alpha_i\in \mathbb R$.  We obtain
\begin{align*}
|D_\lambda| =& \sum_{i=1}^2 |\mu_i||\alpha_i\phi+k_i\bar z \psi\>\<\alpha_i\phi+k_i\bar z \psi|\\
=& \sum_{i=1}^2 |\mu_i| (\alpha_i^2P_\phi+k_i^2 |z|^2P_\psi \\
&+\alpha_ik_iz |\phi\>\<\psi|+\alpha_ik_i\bar z |\psi\>\<\phi|)\\
=&(|\mu_1|\alpha_1^2+|\mu_2|\alpha_2^2)P_\phi \\
&+(|\mu_1| k_1^2+|\mu_2| k_2^2)|z|^2P_\psi+(\alpha_1k_1 |\mu_1| \\
&+\alpha_2k_2|\mu_2|)(z|\phi\>\<\psi|+\bar z |\psi\>\<\phi|).
\end{align*}
\end{document}